\documentclass[twocolumn,pra,showpacs]{revtex4}
\usepackage[latin9]{inputenc}
\usepackage{amsmath}
\usepackage{amssymb}
\usepackage{graphicx}

\makeatletter

\@ifundefined{textcolor}{}
{%
 \definecolor{BLACK}{gray}{0}
 \definecolor{WHITE}{gray}{1}
 \definecolor{RED}{rgb}{1,0,0}
 \definecolor{GREEN}{rgb}{0,1,0}
 \definecolor{BLUE}{rgb}{0,0,1}
 \definecolor{CYAN}{cmyk}{1,0,0,0}
 \definecolor{MAGENTA}{cmyk}{0,1,0,0}
 \definecolor{YELLOW}{cmyk}{0,0,1,0}
}

\usepackage{amsfonts}
\usepackage{graphics}
\usepackage{epsfig}\usepackage{amsthm}

\def\identity{\leavevmode\hbox{\small1\kern-3.8pt\normalsize1}}

\newtheorem{theorem}{\bf{Theorem}}

\newtheorem{lemma}{\bf{Lemma}}

\newcommand{\ket}[1]{\left | #1 \right\rangle}
\newcommand{\bra}[1]{\left \langle #1 \right |}

\newcommand{\Tr}{\mathrm{Tr}}

\renewcommand{\epsilon}{\varepsilon}

\makeatother

\begin{document}

\title{Excessive distribution of quantum entanglement}

\author{Margherita Zuppardo}

\author{Tanjung Krisnanda}

\author{Tomasz Paterek}
\email{tomasz.paterek@ntu.edu.sg}

\affiliation{School of Physical and Mathematical Sciences, Nanyang Technological
University, Singapore}

\affiliation{Centre for Quantum Technologies, National University of Singapore,
Singapore}

\author{Somshubhro Bandyopadhyay}
\email{som.s.bandyopadhyay@gmail.com}

\author{Anindita Banerjee }

\author{Prasenjit Deb }

\author{Saronath Halder}

\affiliation{Department of Physics and Center for Astroparticle Physics and Space
Science, Bose Institute, Block EN, Sector V, Bidhan Nagar, Kolkata
700091, India}

\author{Kavan Modi}

\affiliation{School of Physics, Monash University, Victoria 3800, Australia }

\author{Mauro Paternostro }

\affiliation{School of Mathematics and Physics, Queen's University, Belfast BT7
1NN, United Kingdom}

\pacs{03.65.Ud}
\begin{abstract}
We classify protocols of entanglement distribution as excessive
and non-excessive ones. In a non-excessive protocol, the gain of entanglement
is bounded by the amount of entanglement being communicated between the remote parties, while excessive protocols
violate such bound. We first present examples of excessive protocols
that achieve a significant entanglement gain. Next we consider their
use in noisy scenarios, showing that they improve entanglement achieved
in other ways and for some situations excessive distribution is the
only possibility of gaining entanglement. 
\end{abstract}
\maketitle

\section{Introduction}

Quantum entanglement is not only an essential concept of quantum mechanics
but also ``a new resource as real as energy''~\cite{RevModPhys.81.865}.
Distributing entanglement between two distant laboratories is crucial
for quantum information processing as exemplified by cryptography~\cite{cryptography},
dense coding~\cite{dense-coding} or teleportation~\cite{teleportation}.
Nonetheless limits on entanglement distribution are only recently
studied and not fully understood. 

Remarkably, the pre-availability of entanglement is not necessary to create an 
entangled network of local nodes~\cite{cubitt}. 
This finding has attracted considerable attention, resulting in  
several theoretical proposals ~\cite{mista1, mista2, mista3, sep2} and inspiring test-bed
experimental realizations~\cite{exp0,exp1,exp2,exp3}.  On the other 
hand, it generated great curiosity on what really limits the distribution,
if not the carried entanglement. In  ~\cite{bounds1, bounds2} it has 
been shown that quantum discord is a necessary condition for a successful
distribution, providing an upper bound to the amount of entanglement 
generated. However, the presence of discord in the carrier is not a 
sufficient condition, and e.g. Ref.~\cite{Streltsov1} investigates 
other limitations of the resources, highlighting a link to the 
dimensionality and the rank of the state of the carrier system for 
distribution with separable states. A recent work ~\cite{Streltsov2} 
further investigates the role of carried entanglement and other quantum 
correlations in the presence of noise. In this case, it is shown that 
the optimal strategy may depend on the entanglement measure at hand.

Our work aims to contribute along similar directions by providing a systematic 
characterisation of entanglement-distribution protocols and illustrating explicit examples
where such schemes could be useful.

\begin{figure*}
\includegraphics[width=0.75\textwidth]{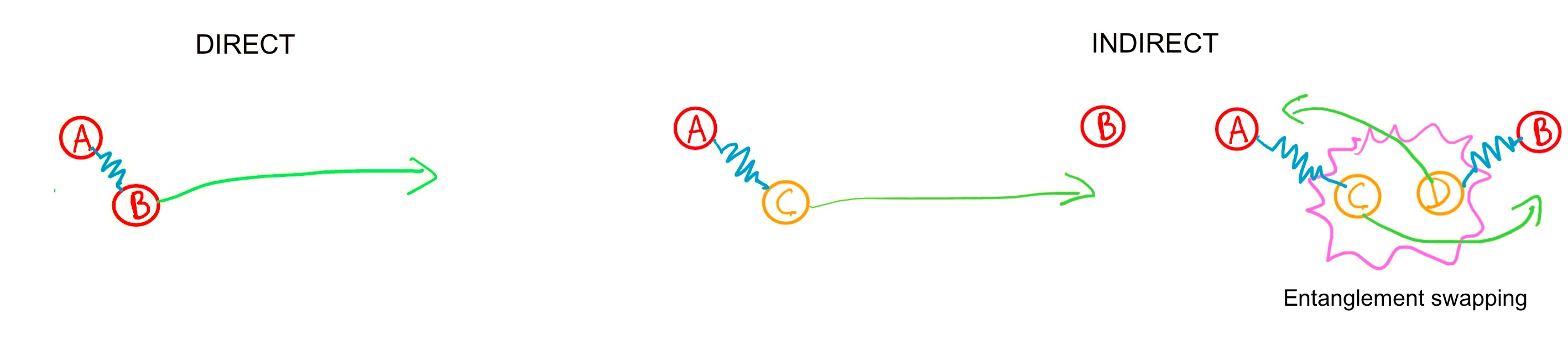}
\caption{(Color online) Direct and indirect protocols for entanglement distribution.
In the direct protocol, the systems of interest get entangled via mutual
interaction. In the indirect protocol they get entangled via interactions
with ancillary systems. As examples, we draw a uni-directional protocol
with ancilla travelling from one lab to the other and an entanglement
swapping scheme with Bell measurement conducted on the ancillae.}
\label{FIG_protocols} 
\end{figure*}

We present two classes of distribution protocols: in \emph{direct} distribution schemes, 
entanglement is first established {\it locally} in one laboratory by direct interaction of two 
subsystems, and then distributed by sending one of them to a remote node of a network. 
In \emph{indirect} schemes, one starts with the subsystems already apart and
uses ancillary systems as communication channels between the laboratories to establish entanglement 
between them. This class encompasses, among others, the intriguing protocols that rely only on separable ancillary carriers mentioned above~\cite{cubitt,mista1, mista2, mista3, sep2,exp0,exp1,exp2,exp3}. 
They reveal that, even if no entanglement is being communicated, i.e.
ancillary systems are at all times separable from the core
subsystems, entanglement in the state of the latter can grow.
The existence of such protocols motivates the study on the second class of protocols.
 We call a protocol \emph{excessive} if entanglement gain exceeds the amount of communicated
entanglement, and \emph{non-excessive} otherwise. Only indirect protocols
can be excessive as in direct distribution schemes the carrier is one
of the core subsystems.

In what follows we provide examples of protocols that belong to each
category, as well as illustrate the relations that exist between them. Our study addresses both
pure and mixed states, in both ideal and noisy conditions.
In Sec.~\ref{SEC_PROTOCOLS} we describe in detail the classifications
above. Sec.~\ref{noiseless} demonstrates that, for certain measures of entanglement,
excessive distribution is possible even with pure states. We also
provide cases of excessive protocols in which considerable entanglement gain
is achieved both via separable and entangled carriers. In Sec.~\ref{noise}
we investigate a possible advantage of indirect and excessive protocols
in noisy environments. It is shown that entanglement distribution
via separable carriers is not possible for entanglement breaking communication
channels. For other channels, excessive protocols are shown to be
useful, especially when there is not much control over the {\it quality} of the state of the ancilla.
Quite remarkably, they are often the only way of increasing entanglement.


\section{Entanglement distribution protocols}

\label{SEC_PROTOCOLS}

We start by describing in detail the two classes of
entanglement distribution protocols addressed in this work and give examples of typical members
of each class.

\subsection{Direct and indirect protocols}

The difference between direct and indirect protocols is illustrated in Fig.\ref{FIG_protocols}.
Direct protocols are the most straightforward ways of increasing
entanglement between distant laboratories: all they entail is the preparation of entangled
states in one laboratory and the transmission of one subsystem to a distant laboratory. 

In an indirect protocol, on the other hand, one requires the use of additional systems to entangle the main
ones, which are typically already located in distant laboratories. The
simplest example is to first entangle an ancilla with a particle present in one of the remote laboratories, and 
then transmit it to the other lab. The protocol would be completed by swapping the state of the ancilla and that of the 
local particle at such remote laboratory. A well-known  example of this
kind of protocol is entanglement swapping~\cite{swapping}, where
entanglement initially present in the state of two system-ancilla pairs is
teleported to the systems alone~\cite{exp_swap1,exp_swap2}.

Under small experimental imperfections and low-noise channels,
the implementation of indirect distribution protocols is likely too demanding to be practically
useful. However, for imperfect operations and situations where noise
cannot be ignored, it is quite natural to conceive that an indirect protocol might be more advantageous 
than direct distribution schemes. Such intuition is reinforced by the results presented in this paper. 

We will focus on simple indirect protocols with ancillae transmitted in one way, i.e. from one laboratory to the other. We will
calculate the final inter-laboratory entanglement achieved through the implementation of a given protocol, thus 
including the core system and the ancillae, instead of focusing on the entanglement between the core subsystems only. 
Although, in principle, these are two different quantities, it has been proven in Ref.~\cite{bounds2}
that entanglement can be localised into the state of the core system as long
as certain dimensionality conditions are satisfied, which hold for
most of the cases discussed in this paper.

\subsection{Excessive and non-excessive protocols}

Our second classification divides the protocols with respect to the
amount of entanglement gained as compared to the communicated entanglement.
Fig.~\ref{FIG_p1} presents quantities relevant to this classification.
We consider the change of entanglement between the laboratories of Alice
and Bob caused solely by the exchange of an ancillary carrier system between them (particle $C$ in Fig.~\ref{FIG_p1}). We define \emph{communicated
entanglement} $E_{\mathrm{com}}$ as the entanglement between $C$ and the systems at the laboratories, which are dubbed $A$ and $B$. Therefore, we take $E_{\mathrm{com}}=E_{AB:C}$.
The change of entanglement resulting from such a communication step is intended as the difference between the inter-laboratory entanglement  
when $C$ is with Bob and when $C$ is with Alice, i.e.
\begin{equation}
\Delta E\equiv E_{A:CB}-E_{AC:B}.
\end{equation}
A protocol is called excessive or non-excessive depending on how $\Delta E$
compares with $E_{\mathrm{com}}$. In particular, we have 
\begin{eqnarray*}
\Delta E & \le & E_{\mathrm{com}}: \qquad\textrm{Non-excessive protocol},\\
\Delta E & > & E_{\mathrm{com}}:\qquad\textrm{Excessive protocol}.
\end{eqnarray*}
Therefore, in an excessive protocol the entanglement gain exceeds the limit set by the communicated entanglement. As we will show, this property is dependent on the choice of entanglement monotone.
\begin{figure}[t]
\includegraphics[width=0.5\textwidth]{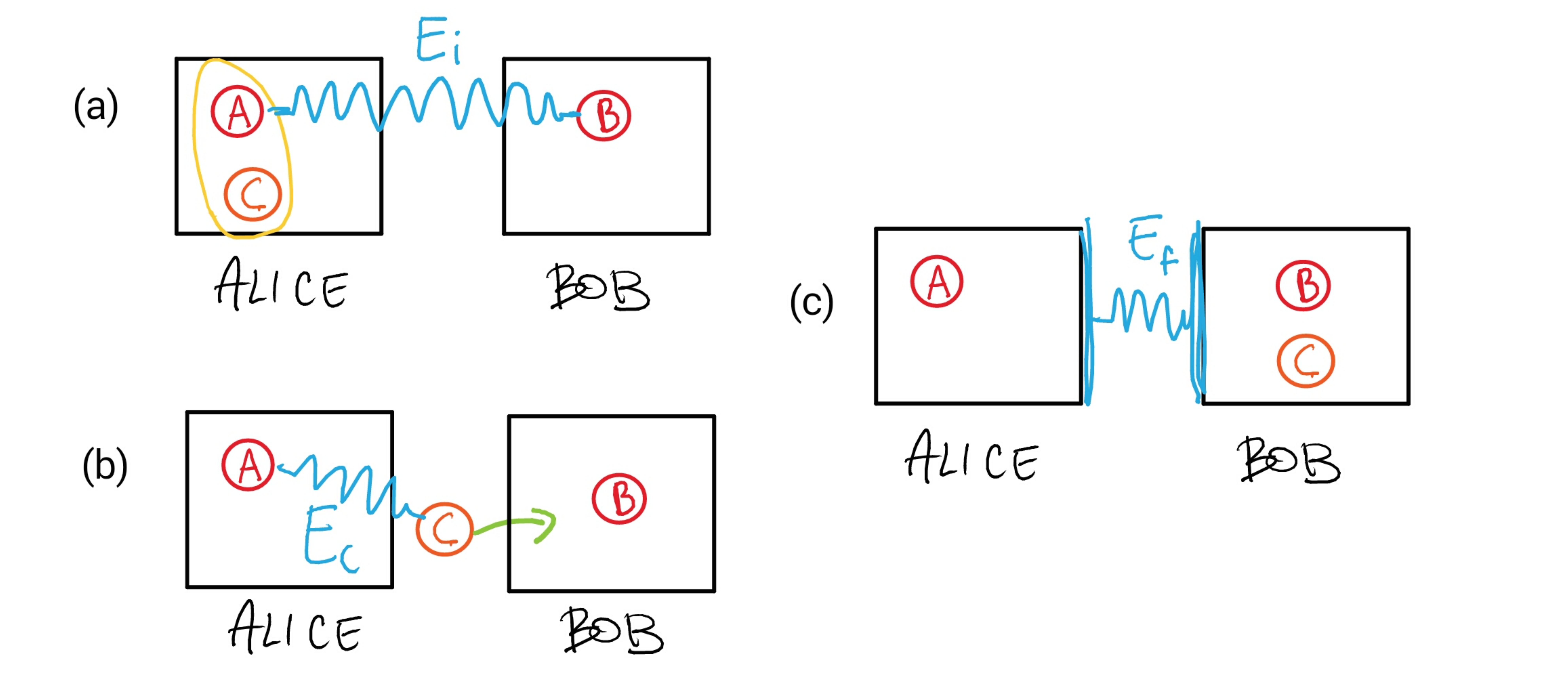} 
\caption{(Color online) Scenario of excessive and non-excessive entanglement
distribution protocols. We study entanglement gain between the laboratories
caused solely by the communication of particle $C$ from Alice to
Bob. (a) Initially entanglement between the laboratories is given by $E_{AC:B}$.
(b) The communicated entanglement is taken to be the entanglement between
the carrier and the other subsystems, i.e. $E_{AB:C}$. (c) Finally the entanglement
between the laboratories is given by $E_{A:CB}$. The protocol is excessive
if $E_{A:CB}-E_{AC:B}>E_{AB:C}$, meaning that the gain is greater
than what was communicated. Otherwise we call it non-excessive.}
\label{FIG_p1} 
\end{figure}

The examples of indirect protocols given above are all non-excessive. Other notable classes of non-excessive
indirect protocols will be presented later. The existence of excessive
protocols was first pointed out in Ref.~\cite{cubitt} and further
examples were presented in Refs.~\cite{mista1, mista2, mista3,sep2}.
However, to the best of our knowledge, an analytical example of excessive protocol with non-zero
communicated entanglement will be presented here for the first time.

In addition to the entanglement change between laboratories, it is also
interesting to investigate the change of entanglement in the principal system.
In Fig.~\ref{FIG_p1}, the latter is composed of particles
$A$ and $B$, which are stationary in the laboratories of Alice and Bob.
As initial entanglement in the principal system $E_{i}$ we naturally
choose the entanglement available before Alice makes her particle $A$ interact with the ancilla $C$.
Assuming that the ancilla is initially uncorrelated from the particles $AB$
(an assumption that holds in all our examples and that is typically verified experimentally) we find that
$E_{AC:B}\le E_{i}$. The final entanglement is the one available after $C$
reaches Bob's laboratory and he applies a general local transformations on his particles
in order to localise the entanglement established between the laboratories into the state of the principal
system, i.e., $E_{f}\le E_{A:CB}$. We therefore conclude that 
\begin{equation}
E_{f}-E_{i} \le E_{A:CB}-E_{AC:B}=E_{\mathrm{fin}}-E_{\mathrm{in}}=\Delta E.
\end{equation}
In our nomenclature $E_f$ and  $E_i$ indicate respectively the final and initial entanglement between the principal subsystems $A$ and $B$,
whereas $E_{\mathrm{fin}}$ and  $E_{\mathrm{in}}$ denote the final and initial entanglement between the laboratories, respectively.
Non-excessive protocols for entanglement between the laboratories are also
non-excessive for the principal system, but the excessive protocols for laboratories do
not guarantee that the principal system gains entanglement above the communicated one. 
This is because not all entanglement can be localised
into the principal system and some of its initial entanglement might
be destroyed by interactions with the ancilla.


\section{Ideal conditions}

\label{noiseless}

In this Section we investigate entanglement gain for the ideal case
where there is no noise in the communication channel between laboratories. We first study the indirect
protocol of Fig. \ref{FIG_p1}, where the state of the whole $ABC$ system is pure. It turns out that the excessiveness depends on the
particular entanglement measure being used. We then present a single parameter
family of five-qubit states that are shown to provide a single platform
exhibiting various possibilities of entanglement gain.

\subsection{Entanglement measures }

For sub-additive measures proportional to the entropy of subsystems, like the von
Neumann entropy and the linear entropy for pure states, we find
that the protocols are always non-excessive. For other measures, such
as negativity~\cite{negativity} and logarithmic negativity~\cite{negativity,logneg1}, we show that even pure states
of sufficiently high dimension give rise to excessive gain.

\subsubsection{Sub-additive measures}

We begin by noting that for pure states the condition for
non-excessiveness is equivalent to the sub-additivity of measures
that characterise pure state entanglement in terms of properties of
a subsystem. The non-excessiveness condition reads
\begin{equation}
E_{A:CB}\le E_{AC:B}+E_{AB:C}.
\end{equation}
If the entanglement involved in the equation above embodies a property of a subsystem,
such as $E_{i:jk}=S_{jk}$ with $S_{jk}$ being a property of subsystem $jk$ (here $i,j,k=A,B,C$),
we can rewrite this as 
\begin{equation}
S_{BC}\le S_{B}+S_{C}.
\label{subadditivity}
\end{equation}
This is exactly the sub-additivity property, which appears as a necessary and sufficient consequence of the non-excessive nature of a protocol. 

\subsubsection{Negativity}

We move to computable entanglement as characterised by negativity~\cite{negativity} and show that all protocols in which entanglement is measured by this
quantity are non-excessive if the dimension of $A$ is less than 3. We provide a simple
explicit example of a state in $3\times2\times2$ dimension which allows for excessive entanglement gain.

We first prove a lemma that reveals dimensionality dependent inequality for the negativity which we will then exploit.
Recall that negativity of a bipartite state $\rho_{XY}$ is defined as
 \begin{equation}
N_{X:Y}=\frac{\|{\rho_{XY}^{PT}}\|-1}{2},
\label{negativity}
 \end{equation}
where PT indicates the partial transposition with respect to subsystem $X$ and $\|\sigma\|=\Tr\sqrt{\sigma^\dag\sigma}$ denotes the trace norm of a generic operator $\sigma$.

\begin{lemma}
\label{LEM_NEG}
The following negativity inequality holds for arbitrary tripartite system in a pure state:
\begin{equation}
\label{LEM_DNEG}
\sqrt{\tfrac{2}{d_A(d_A-1)}} N_{A:CB} \le N_{AC:B} + N_{AB:C},
\end{equation}
where $d_A$ is the rank of the reduced state of Alice.
\end{lemma}
\begin{proof}
Let us write the global pure state $\ket{\psi}$ in its Schmidt forms for various bipartitions:
\begin{equation}
\label{schmidt}
\begin{aligned}
\ket{\psi} &=\sum_{\alpha=1}^{d_A}{\sqrt{p_{\alpha}}\ket{\alpha}_{A}\ket{\phi_{\alpha}}_{BC}}\\
&=\sum_{\beta=1}^{d_B}{\sqrt{q_{\beta}}\ket{\beta}_{B}\ket{\chi_{\beta}}_{AC}}\\
&=\sum_{\gamma=1}^{d_C}{\sqrt{r_{\gamma}}\ket{\gamma}_{C}\ket{\xi_{\gamma}}_{AB}},
\end{aligned}
 \end{equation}
where we have introduced suitable orthonormal bases.
In this notation, the respective negativities are:
\begin{equation}
\label{schmidtneg}
\begin{aligned}
N_{A:CB}&=\frac{1}{2}\sum_{\alpha \neq a}{\sqrt{p_\alpha p_a}},\\
N_{AC:B}&=\frac{1}{2}\sum_{\beta\neq b}{\sqrt{q_\beta q_b}},\\
N_{AB:C}&=\frac{1}{2}\sum_{\gamma \neq c}{\sqrt{r_\gamma r_c}}.
\end{aligned}
\end{equation}
Our starting point is the sub-additivity of linear entropy~\cite{PhysRevA.75.052324}, which in present notation reads:
\begin{equation}
\sum_{\alpha \neq a}{p_\alpha p_a}\ \le \sum_{\beta \neq b}{q_\beta q_b} + \sum_{\gamma \neq c }{r_\gamma r_c}.
\label{S}
\end{equation}
We obtain the lower bound on the left-hand side by noting that the sum can be interpreted as the length of vector $(\sqrt{p_1 p_2}, \sqrt{p_1 p_3}, \dots,\sqrt{p_{d_A-1} p_{d_A}})$,
whereas its inner product with vector $(\frac{1}{2}, \frac{1}{2},\dots,\frac{1}{2})$ gives the negativity $N_{A:CB}$.
By the Cauchy-Schwatrz inequality the lower bound is
\begin{equation}
\label{2}
\frac{4}{d_A(d_A-1)} N_{A:CB}^2 \le \sum_{\alpha \neq a}{p_\alpha p_a}.
\end{equation}
For the upper bound to the right-hand side of (\ref{S}) consider
\begin{equation}
\label{1}
\begin{aligned}
(N_{AC:B}+N_{AB:C})^2 &\ge N_{AC:B}^2+N_{AB:C}^2 \\
&= \frac{1}{4} \sum_{\beta \neq b}{\sqrt{q_\beta q_b}}\sum_{{\beta\prime} \neq {b \prime}}{\sqrt{q_{\beta\prime} q_{b \prime}}} \\
&+ \frac14 \sum_{\gamma \neq c}{\sqrt{r_\gamma r_c}} \sum_{\gamma\prime \neq c \prime}{\sqrt{r_{\gamma\prime} r_{c \prime}}}\\
&\ge \frac{1}{2} \sum_{\beta \neq b}{q_\beta q_b}+\frac{1}{2} \sum_{\gamma \neq c}{r_\gamma r_c}.
\end{aligned}
\end{equation}
The last inequality holds due to the fact that, in the above sums, the combination of two pairs of equal indexes occurs twice,
e.g. we get $q_\beta q_b$ by multiplying terms with $\beta=\beta\prime$ and $b = b \prime$ and also $\beta = b \prime$ and $b =\beta\prime$. 
All the remaining terms are positive, hence the inequality.
By combining the lower bound and the upper bound we arrive at inequality~(\ref{LEM_DNEG}).
\end{proof}
We are now ready to make our main statement about excessiveness in terms of negativity.

\begin{theorem} \label{TH_NEG_INEQ} 
The following inequality holds for all pure states if and only if subsystem $A$ is a qubit:
\begin{equation}
\label{TH_NEG}
N_{A:CB} \le N_{AC:B} + N_{AB:C}.
\end{equation}

\end{theorem}
\begin{proof} 
Using (\ref{LEM_DNEG}) with $d_A = 2$ we find exactly the non-excessiveness condition.
If $d_A > 2$ we provide an explicit \emph{minimal} example (in terms of size of the subsystems) of a state that leads to excessive protocol.
Choose $d_A = 3$, $d_B = d_C = 2$, and consider the state
\begin{equation}
\ket{\psi}=\frac{1}{\sqrt{3}}(\ket{200}+\ket{001}+\ket{110}),
\label{violation}
\end{equation}
for which the communicated negativity is given by ${N_{AB:C}=}\sqrt{2}/3\approx0.471$ while the negativity gain is ${N_{A:CB}-N_{AC:B}}=1-\sqrt{2}/3\approx0.529$.
\end{proof}

Theorem \ref{TH_NEG_INEQ} has a consequence for  the ongoing effort aimed at unifying the current approaches to quantum correlations~\cite{modi2010,giorgi2011,modi-vedral}. 
Their goal in this respect is to quantify general quantum correlations (including entanglement, quantum discord, etc.) with the same mathematical forms, thus allowing for direct comparison of their respective values. 
We argue that discord-like quantity on equal footing with negativity will not satisfy physically plausible properties.
Namely, there are two properties that we would expect from a unified approach:
(i) Since all non-classical correlations of pure states should be due to quantum entanglement, in a unified approach discord-like quantity should reduce to negativity for pure states;
(ii) We would expect the discord-like quantity to measure non-classicality of communication as it is the case for other measures~\cite{bounds1, bounds2}, and therefore in a unified approach the discord-like quantity should bound the negativity gain. However, the violation of inequality~\eqref{TH_NEG} shown by higher-dimensional systems implies that there cannot be a discord-like quantity that reduces
to negativity for pure states and also respects condition (ii).

\subsubsection{Logarithmic negativity}

A measure related to negativity is the {\it logarithmic negativity}, defined as~\cite{negativity,logneg1}
\begin{equation}
L_{X:Y}=\log_2\|\rho^{PT}_{XY}\|=\log_2(2N_{X:Y}+1).
\end{equation}

Similarly to what has been discussed above, we can identify excessive protocols based on the use of logarithmic negativity and a system $A$ of sufficiently high dimensionality. 
\begin{theorem} \label{TH_LNEG_INEQ} For pure states of three subsystems $A$, $B$, and $C$ with $d_A=2$ we have
\begin{equation}
\label{TH_L}
L_{A:CB}-L_{AC:B}\le L_{AB:C}.
\end{equation}
\end{theorem} \begin{proof} The proof follows trivially from
inequality~(\ref{TH_NEG}). Multiply inequality~(\ref{TH_NEG}) by two and add one to both sides. Taking the logarithm gives us  
\begin{equation}\label{2nn}
\log_2(2N_{A:CB}+1)\le\log_2(2N_{AC:B}+2N_{AB:C}+1).
\end{equation}
The thesis follows if we combine $\log(2N_{AC:B}+2N_{AB:C}+1)\le\log(2N_{AC:B}+1)+\log(2N_{AB:C}+1)$ with inequality~(\ref{2nn}). \end{proof}

An extensive numerical analysis performed by testing uniformly picked random pure states suggests that, differently from what has been found for the negativity inequality~(\ref{TH_NEG}), inequality~(\ref{TH_L}) always holds for $d_A=3$. The first example of a pure state that does not satisfy inequality~(\ref{TH_L})
has been encountered for $d_A=4$ with $B$ and $C$ both being qubits. For instance, for state
\begin{equation}
\ket{\psi}=\frac{1}{\sqrt{103}}\left(10\ket{000}+\ket{110}+\ket{201}+\ket{311}\right),
\end{equation}
the communicated logarithmic negativity is $L_{AB:C}
\approx 0.352$ while the corresponding gain in logarithmic negativity is $\approx 0.363$, which is excessive. 
Similar to violations in the case of negativity, this happens due to the dimension of $A$. 
We conjecture that inequality~(\ref{TH_L}) holds in Hilbert spaces of arbitrary dimension where subsystem $A$ has dimension less than $4$.

\subsection{Excessive protocols}

We now move to a single-parameter family of states which allows for
various possibilities of entanglement gain. We will emphasise excessive
protocols as they are our main focus here. Recall that a protocol
is said to be excessive if $\Delta E>E_{\mathrm{com}}$, where $E_{\mathrm{com}}=E_{C:AB}$
is the entanglement of the carrier with the rest and
$\Delta E=E_{\mathrm{fin}}-E_{\mathrm{in}}$, with $E_{\mathrm{in}}=E_{AC:B}$
and $E_{\mathrm{fin}}=E_{A:BC}$ respectively being the initial and final entanglement between Alice's and Bob's laboratories.

Notice that an excessive protocol can, in principle, be realised for
all four possible scenarios that correspond to whether the initial and/or communicated entanglement are vanishing. 
The first example of an excessive protocol demonstrating the possibility to distribute entanglement via a separable carrier state 
 had $E_{\mathrm{in}}=0$ and $E_{\mathrm{com}}=0$~\cite{cubitt}, whereas in Ref.~\cite{bounds2}
an example of an excessive protocol was given with $E_{\mathrm{in}}>0$ and $E_{\mathrm{com}}=0$. 
These results may also be understood as
a direct consequence of the fact that a tripartite density matrix
$\rho_{ABC}$ can have different entanglement across different bipartitions.
For example, in the entanglement distribution via separable states
scenario one has $E_{AC:B}=E_{C:AB}=0$, yet $E_{A:BC}>0$.
In both the above examples the communicated entanglement is zero as
the carrier $C$ remains unentangled, at all times, with other subsystems. We will provide here for the first time examples of
excessive protocols with non-zero communicated entanglement and
give new examples for the scenario where $E_{\mathrm{in}}>0$ and $E_{\mathrm{com}}=0$.
It should be noted that all such scenarios can be encompassed by a single-parameter family of density matrices, as will be shown bellow.

We also observe that the excessive nature of a protocol may depend on whether an intermediate particle (which does not influence the initial entanglement) is sent in advance of the designated carrier. 
We call this strategy \emph{catalysis of excessiveness}. It can be described as follows. 
Suppose a protocol $P$ is excessive over a certain range $\Delta$ of a given parameter
associated with the state (it is not important to specify which parameter). Let the carrier be denoted by $C$.
Now consider another protocol $P^{\prime}$, where an intermediate
particle $C^{\prime}$ is transmitted without changing the initial
entanglement. This implies that the initial entanglement $E_{{\rm in}}$
is the same for both protocols before the transmission of the designated carrier
$C$. Let $\Delta^{\prime}$ be the range over which $P^{\prime}$
is excessive. We find that not only $P^{\prime}$ is excessive over
a wider range, that is $\Delta^{\prime}>\Delta$, but also that, within
the range $\Delta$ where $P$ and $P^{\prime}$ are both excessive, the
entanglement gain in $P^{\prime}$ is greater than the corresponding one in $P$. Thus an intermediate carrier,
which does not change the initial entanglement, can make an excessive protocol better. We will give
explicit examples demonstrating this effect later.  

Consider a five qubit density matrix obtained by applying local
quantum channels on a five-qubit {\it absolute maximally entangled state} (AME)
\cite{AME}. AME states have the property
to be maximally entangled across every bipartition. Thus we
might expect such states to be more robust to local noise and
conceivably good candidates to exhibit excessive protocols. The five-qubit pure AME state is given by
\begin{eqnarray}
\vert\psi\rangle & = & \frac{1}{4}(\vert00000\rangle+\vert10010\rangle+\vert01001\rangle+\vert10100\rangle\nonumber \\
 &  & +\vert01010\rangle-\vert11011\rangle-\vert00110\rangle-\vert11000\rangle\nonumber \\
 &  & +\vert11101\rangle-\vert00011\rangle-\vert11110\rangle-\vert01111\rangle\nonumber \\
 &  & +\vert10001\rangle-\vert01100\rangle-\vert10111\rangle+\vert00101\rangle).\label{AME-state}
\end{eqnarray}
We then construct the density matrix $\rho\left(q\right)$ resulting from the application of two specific local quantum channels on the first two qubits of such state. Such channels, which we label $\Lambda_{1,2}~(k=1,2)$, are defined in
terms of their respective Kraus operators as \begin{equation*}
\begin{aligned}
K_{0}^{(1)}&=\frac{1}{\sqrt2}\,\openone,\quad K_{i}^{(1)}=\frac{1}{\sqrt6}\,\sigma_{i},\\
K_{0}^{(2)}&=\sqrt{q}\,\openone,  \quad K_{i}^{(2)}=\sqrt{\frac{1-q}{3}}\,\sigma_{i}.
\end{aligned}
\end{equation*}
for $q\in[0,1]$ and $i=x,y,z$. Correspondingly, the five-qubit state is now 
\begin{equation}
\rho\left(q\right)  =  \left(\Lambda_{1}\otimes\Lambda_{2}\otimes\mathcal{I}_{3}\otimes\mathcal{I}_{4}\otimes\mathcal{I}_{5}\right)[\vert\psi\rangle\langle\psi\vert].\label{EQ_RHO}
\end{equation}
Channel $\Lambda_{1}$ is always entanglement breaking, and thus
the first qubit becomes unentangled with the rest of the system. Channel $\Lambda_{2}$ is entanglement breaking for $0\leq q\leq0.5$. The application of such channels breaks the symmetry initially present in the AME state across its bipartitions. The excessive nature of the entanglement distribution protocols that we are going to describe is then the consequence of such breakdown of symmetry. Table~\ref{TAB_SEP} summarizes the separability properties of the state across the relevant bipartitions.

\begin{table}
\begin{tabular}{|c | c | c|}
\hline 
\hspace{.25cm} Partition \hspace{.25cm} & \hspace{.5cm} $0 \le q \le 0.5$ \hspace{.25cm} & 
\hspace{.25cm} $0.5 < q \le 1$ \hspace{.25cm} \\
\hline \hline 
$12:345$  & separable  & entangled  \\
\hline 
$2:1345$  & separable  & entangled  \\
\hline 
$1:2345$  & \multicolumn{2}{c|}{separable} \\
\hline 
$3:1245$  & \multicolumn{2}{c|}{entangled} \\
\hline 
$13:245$  & \multicolumn{2}{c|}{entangled} \\
\hline 
$123:45$  & \multicolumn{2}{c|}{entangled} \\
\hline 
\end{tabular}\protect\caption{Separability properties of the five-qubit state $\rho(q)$ defined
in Eq.~(\ref{EQ_RHO}). The first column specifies the relevant bipartitions
of the five qubits whereas the second column indicates whether in
the corresponding bipartition the state is entangled. For the top two bipartitions
separability changes as parameter $q$ is tuned above $0.5$.}
\label{TAB_SEP}
\end{table}

Having set the resource state to use, we now present new excessive protocols for various combinations of initial and communicated entanglement. 
In each case we give two examples to highlight the fact that, given a
system of many particles in some specified quantum state, there could
be different tripartite configurations giving rise to an excessive
protocol of the \emph{same kind}.

\subsubsection{Excessive protocols with $E_{\mathrm{in}}=0$ and $E_{{\rm com}}>0$
\protect \\  ($E_{\mathrm{in}}>0$ and $E_{{\rm com}}=0$)}

\begin{figure}[!t]
\includegraphics[width=\columnwidth]{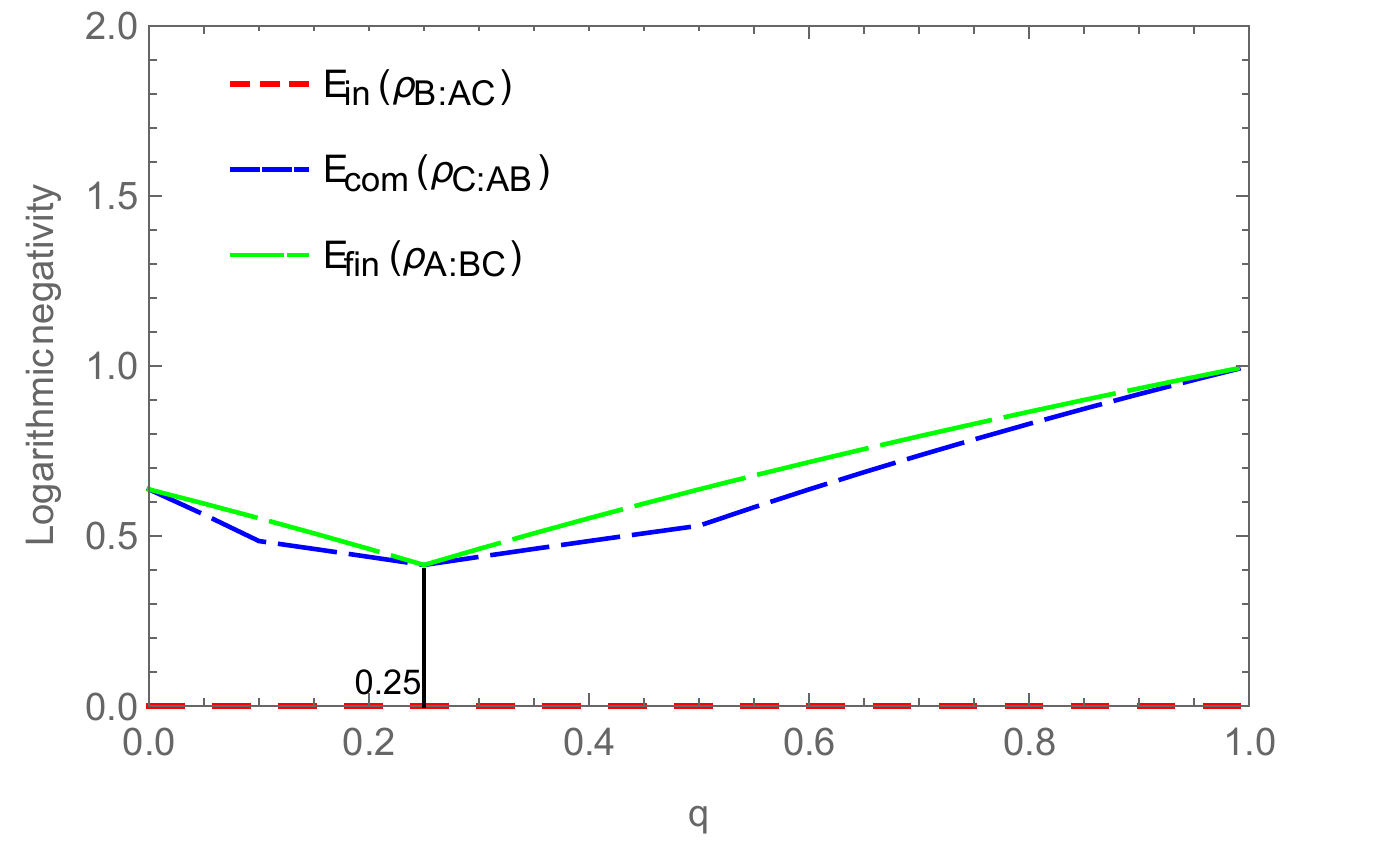}
\caption{(Color online) Excessive protocols with $E_{\mathrm{in}}=0$ and $E_{{\rm com}}>0$ realised through a local channel-affected five-qubit AME state.
For all values of $q$ except at $q=\{0,1/4,1\}$ we see that $E_{\mathrm{fin}}>E_{{\rm com}}$. The partitions used here are $A=\{2,4,5\}$, $B=\{1\}$ and $C=\{3\}$. 
By swapping qubits $1$ and $3$ we obtain excessive
protocols with $E_{\mathrm{in}}>0$ and $E_{{\rm com}}=0$. }
\label{FIG_SOM-1}
\end{figure}

Let us group the five qubits into the following subsystems $A=\{2,4,5\}$,
$B=\{1\}$, and $C=\{3\}$. Alice initially holds subsystems $A$ and
$C$ whereas Bob holds $B$. As seen in Table~\ref{TAB_SEP}, there
is no entanglement between the laboratories in this configuration, i.e. $L_{AC:B}=0~\forall q$.
 However, Fig.~\ref{FIG_SOM-1} demonstrates that sending
$C$ through a noiseless channel generates more final entanglement
than what was communicated (in terms of logarithmic negativity). More specifically 
\begin{equation}\label{LNEG99}
L_{A:BC}-L_{AC:B}>L_{AB:C}\quad\textrm{ for }\quad q\ne \{0,\tfrac{1}{4},1\}.
\end{equation}
The same kind of excessive protocol is obtained for yet another grouping
of qubits $A=\{1,4,5\}$, $B=\{2\}$, and $C=\{3\}$. As seen in Table~\ref{TAB_SEP},
the initial entanglement vanishes for $0\le q\le0.5$, whereas the
carrier particle is entangled for all $q$. Fig.~\ref{SOM_FIG_2} 
reveals that this protocol is thus excessive for $0.4<q\leq0.5$.

A new family of examples with $E_{\mathrm{in}}>0$ and $E_{{\rm com}}=0$
is obtained from the cases studied above by simply exchanging the roles of
subsystems $B$ and $C$. Protocols that were excessive before are still excessive after the swap, as it can be verified
by rewriting inequality~\eqref{LNEG99} as $L_{A:BC}-L_{AB:C}>L_{AC:B}$. By our analysis above, the
swap $B\leftrightarrow C$ also exchanges $E_{\mathrm{in}}$ and $E_{{\rm com}}$.

\begin{figure}
\includegraphics[width=\columnwidth]{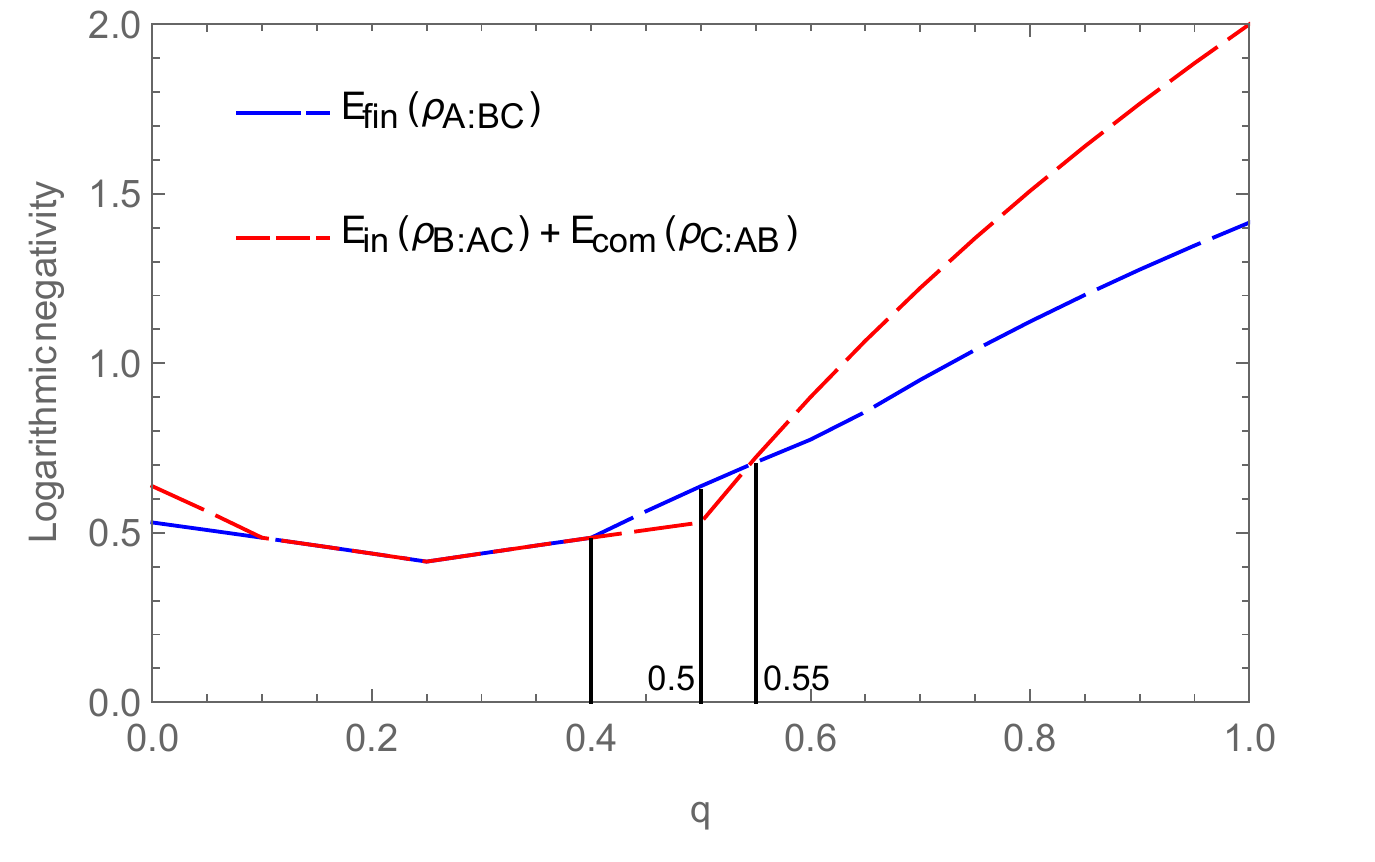}
\caption{(Color online) Excessive protocols with non-vanishing communicated entanglement. 
For $0.4<q\leq0.5$ the protocol is excessive with $E_{\mathrm{in}}=0$ and $E_{{\rm com}}>0$. 
For $0.5<q<0.55$ the protocol is excessive while satisfying $E_{\mathrm{in}}>0$ and $E_{{\rm com}}>0$. The partitions here are $A=\{1,4,5\}$, $B=\{2\}$ and $C=\{3\}$.}
\label{SOM_FIG_2}
\end{figure}

\subsubsection{Excessive protocols with $E_{\mathrm{in}}>0$ and $E_{{\rm com}}>0$}
\label{SEC_EX}

Consider again the grouping $A=\{1,4,5\}$, $B=\{2\}$ and $C=\{3\}$ of the particles in the state in Eq.~\eqref{EQ_RHO}.
This time, we take into account the range  $0.5<q\le1$ in which
the laboratories are initially entangled, i.e. $L_{AC:B}>0$. Under these conditions
the carrier particle is also entangled, i.e. $L_{AB:C}>0$. Fig.~\ref{SOM_FIG_2}
reveals that the protocol is excessive for $0.5<q<0.55$. Similar
conclusions hold under the swap $B\leftrightarrow C$.

\begin{figure}[!t]
\includegraphics[width=\columnwidth]{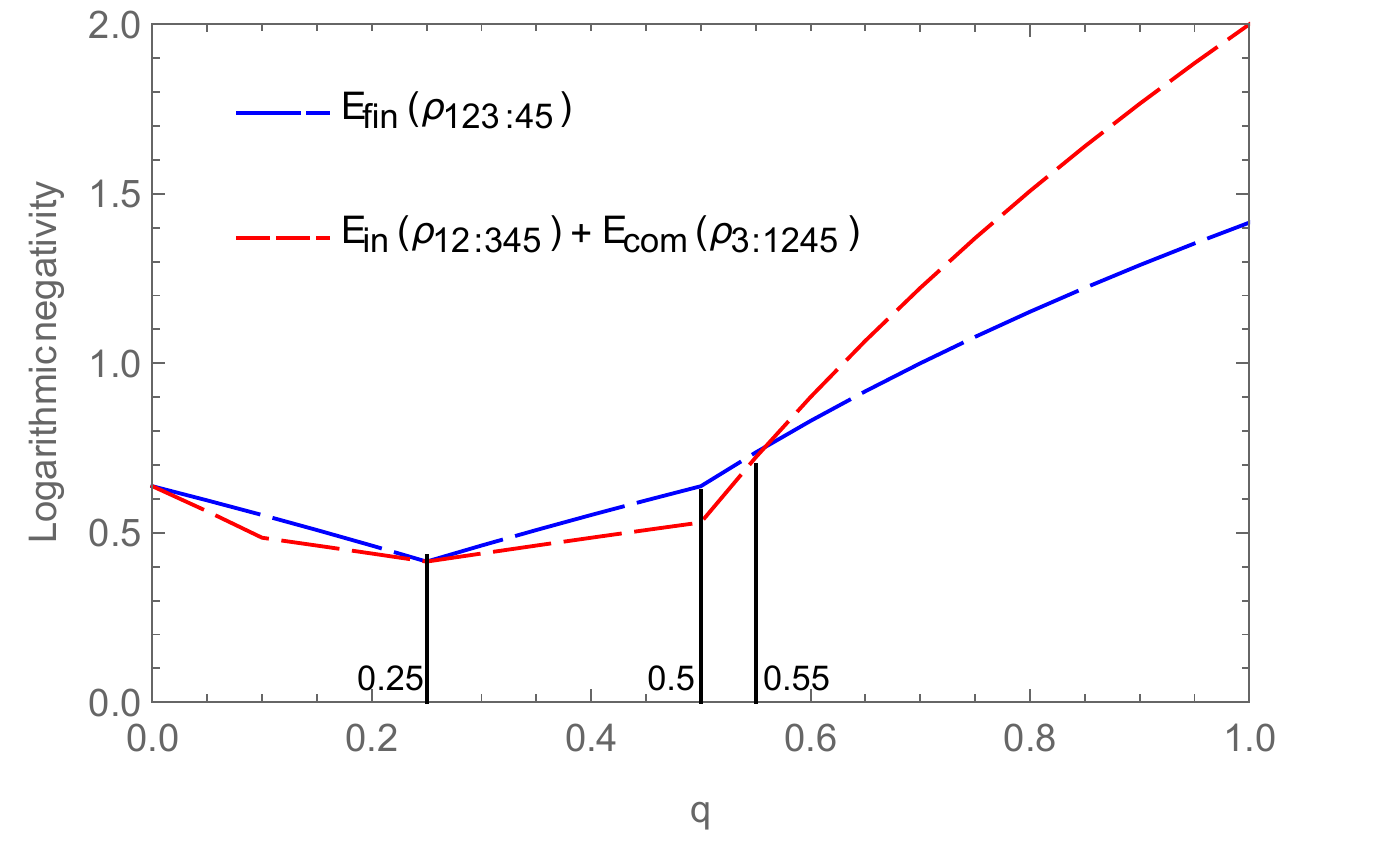}
\caption{(Color online) Catalysis of excessiveness.
Comparison of this plot with Fig.~\ref{SOM_FIG_2} shows that sending qubit $1$ in advance increases entanglement gain between the laboratories and the range of $q$ over which the protocol is excessive
while having no influence on the initial entanglement.}
\label{SOM_FIG_CAT}
\end{figure}

\subsubsection{Catalysis of excessiveness}
A form of catalysis of entanglement is implicitly present in the protocol for distribution with separable states \cite{cubitt}. There, particle $B$ which is separable initially from $A$ and $C$, is first sent to Bob's lab keeping vanishing entanglement between Alice and Bob. Next, Alice sends $C$, while  $E_{AB:C}=0$. However, $E_{A:CB}>0$ in the end. Particle $B$ worked as a \emph{catalyst}. 

Despite the analogies, the phenomenon dubbed here as {\it catalysis of excessiveness} is more general. In order to illustrate this, we discuss another example starting from the state $\rho(q)$ as in Eq.~\eqref{EQ_RHO}. Let us now group the five qubits as  $A = \{4,5\}$, $B = \{1,2\}$ and $C = \{3\}$. This implies that, compared with the example in the previous subsection, qubit $1$ has been sent by Alice to Bob {\it before} the protocol begins.
According to Table~\ref{TAB_SEP}, the initial states in both cases (before and after sending qubit $1$) have the same separability properties 
and one can verify that the actual logarithmic negativities are exactly the same.
Therefore, communication of qubit $1$ has no influence on the initial entanglement between the laboratories.

However, if we now send subsystem $C$ to Bob, we notice different behaviour of entanglement gain.  By computing the amount of initial, communicated, and final entanglement, we find the results displayed in  Fig.~\ref{SOM_FIG_CAT}, which show that the pre-delivery of a qubit to Bob's lab does influence the performance of entanglement distribution. Indeed, by comparing Fig.~\ref{SOM_FIG_2} and \ref{SOM_FIG_CAT}, we see that while in the former the protocol is excessive only for $0.4 < q <0.55$, the latter reveals excessiveness for all $q < 0.55$ (except for $q = \{0, 0.25\}$).
Since $q$ parameterises local noise acting on an individual qubit of the register, this can be regarded as increase in robustness of the protocol.
Furthermore, the actual gain of entanglement in the excessive part of the protocol is larger if qubit $1$ is communicated in advance, thus clarifying the catalytic role of such pre-step in the protocol. 

\section{Noisy environments}

\label{noise}

One of the main obstacles to the generation and preservation of entanglement
is the presence of noisy environments. As full noise-avoidance appears to be too demanding or costly to 
embody a viable way to circumvent the problem, a potential approach to face the challenge of noisy entanglement
distribution is to design protocols able to work well in such non-ideal conditions.

The noise affecting the communication channel that connects two laboratories will interfere mainly with the entanglement communicated between them. This
suggests that excessive protocols, that have a small amount of communicated
entanglement compared to the gain, could work better than non-excessive
ones. In this Section we verify the efficiency of concrete indirect
protocols for entanglement distribution in the presence of three typical
quantum noises. We also address the most extreme case of noisy channels, i.e. entanglement
breaking channels. The results that have been achieved through our analysis suggest that excessive protocols often allow
for significant amounts of distributed entanglement, even under the presence of rather strong noise.
As a quantitative measure of entanglement we use negativity and we focus on the situation where the three subsystems involved in the distribution schemes are all qubits. According to the analysis in Sec.~\ref{noiseless}, the
communicated entanglement always exceeds the entanglement increment
if the three qubit system is in a pure state but, since noise will
necessarily mix the system, excessive distribution becomes possible.

\subsection{Noisy Channels}

\label{channels}

We begin by introducing three standard channels modelling environmental
noise acting on two-level systems: dephasing channel, depolarising
channel, and amplitude damping channel. For a unified description,
all of them are represented in terms of Kraus operators involving
a single parameter that characterizes
the strength of the noise. 

\subsubsection{Entanglement breaking channel}

In general, if a channel produces separable output independent of
the input state, it is said to be entanglement breaking (EB).
As proven in Ref.~\cite{ebc}, a channel is EB if and
only if its action on any state can be written as a measure-and-prepare
positive operator-valued measurement (POVM) on the particle that goes through the channel. In the tripartite
scenario at the core of our work, the channel $\Phi_{C}$ acts on the ancillary particle $C$ that is communicated between the 
laboratories, so that the resulting state can be written as 
\begin{equation}
(\mathcal{I}_{AB}\otimes\Phi_{C})(\rho_{ABC})=\sum_{n}p_{n}\rho_{AB|n}\otimes\gamma_{n},\label{povm}
\end{equation}
where $p_{n}$ are the probabilities associated with the measurement outcomes that are part of the POVM performed on $C$, and $\rho_{AB|n}$ are the  states of $AB$ conditioned to the outcomes of the POVM measurement. Finally, $\gamma_{n}$ are
rank-one projectors (pure states) that one prepares on $C$, depending
on the result of the measurement.

Each of the channels that will be introduced in the following subsections has a critical value of its characteristic noise
parameter above which it becomes EB.

\subsubsection{Dephasing channel}

The dephasing channel captures the loss of coherence in a preferred
basis. The strength of this loss is given by the parameter $\delta_{\mathrm{ph}}$
and, if the preferred basis is chosen to be that embodied by the eigenbasis of the $\sigma_{z}$
Pauli matrix, the Kraus operators of the dephasing channel take the
form 
\begin{equation}
	K_{0}^{(\mathrm{ph})} =  \sqrt{1-\frac{\delta_{\mathrm{ph}}}{2}}\,\openone,\qquad K_{1}^{(\mathrm{ph})}  =  \sqrt{\frac{\delta_{\mathrm{ph}}}{2}}\,\sigma_{z}.
\label{dephasing}
\end{equation}
This channel is EB when $\delta_{\mathrm{ph}}=1$.

\subsubsection{Depolarising channel}

The depolarising channel describes the loss of coherence in any basis.
It is defined by the Kraus operators 
\begin{equation}
K_{0}^{(\mathrm{pol})}  = \sqrt{1-\delta_{\mathrm{pol}}}\,\openone,\qquad K_{x,y,z}^{(\mathrm{pol})} = \sqrt{\frac{\delta_{\mathrm{pol}}}{3}}\,\sigma_{x,y,z}.\label{depolarising}
\end{equation}
The channel is EB for $\delta_{\mathrm{pol}}\in[1/2,1]$.

\subsubsection{Amplitude damping channel}

The amplitude damping channel describes energy dissipation from the
system. The Kraus operators for this channel are 
\begin{equation}
\begin{aligned}
K_{1}^{(\mathrm{ad})} & =  |0\rangle\langle0|+\sqrt{1-\delta_{\mathrm{ad}}}|1\rangle\langle1|,\label{amplitude}\\
K_{2}^{(\mathrm{ad})} &=  \sqrt{\delta_{\mathrm{ad}}}|0\rangle\langle1|.
\end{aligned}
\end{equation}
Similarly to the dephasing channel, the amplitude damping channel
is EB for $\delta_{\mathrm{ad}}=1$. 

\begin{figure}[t]
\includegraphics[width=0.5\textwidth]{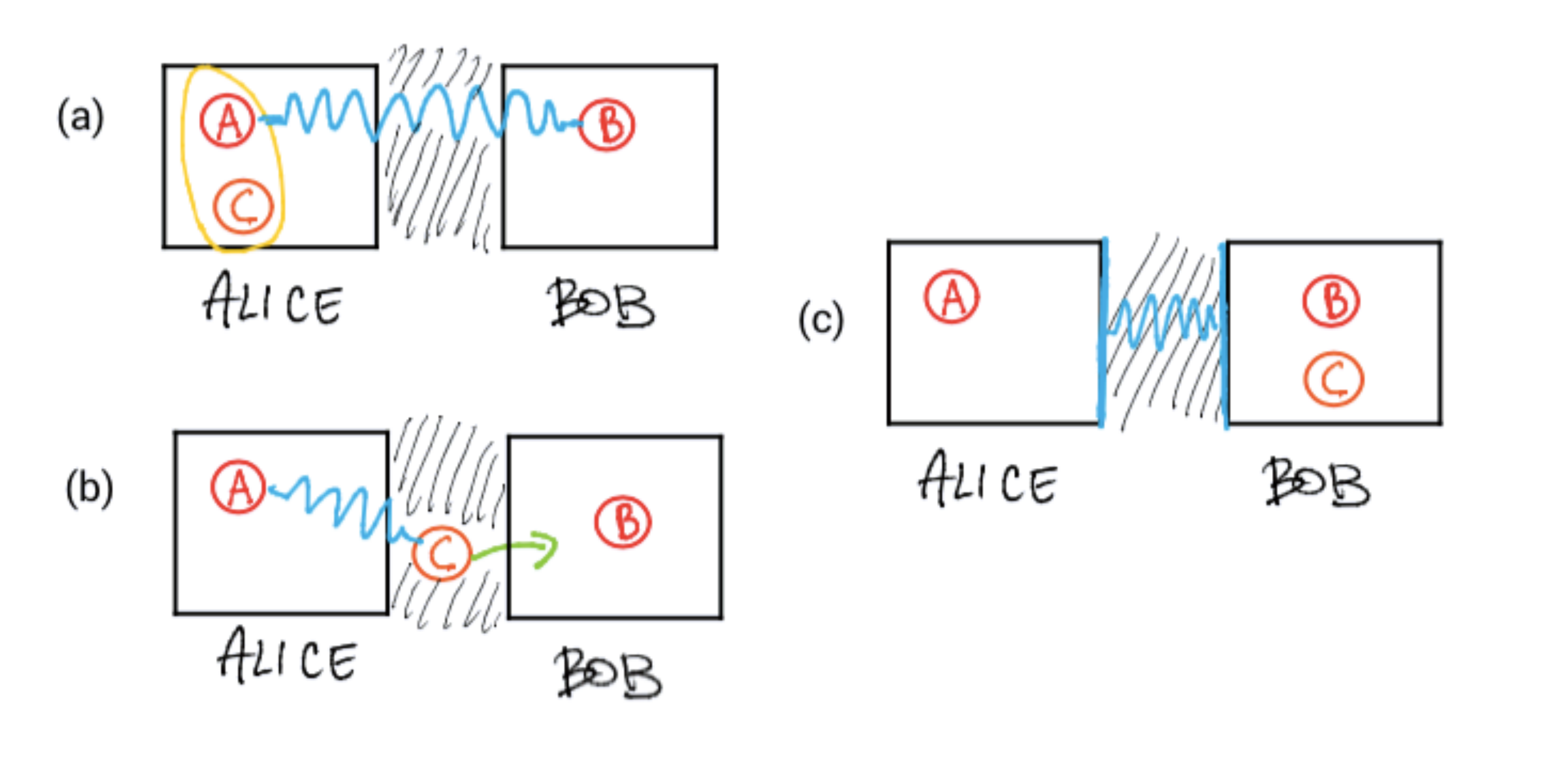}
\caption{(Color online) Indirect entanglement distribution via noisy channel.
(a) $A$ and $B$ are initially already displaced and correlated when
$C$ interacts locally with $A$. (b) $C$ is communicated via noisy
channel to Bob's laboratory. (c) The final entanglement between the
laboratories is $E_{A:BC}$.}
\label{p1} 
\end{figure}

\begin{figure*}[t]
\includegraphics[width=1\textwidth]{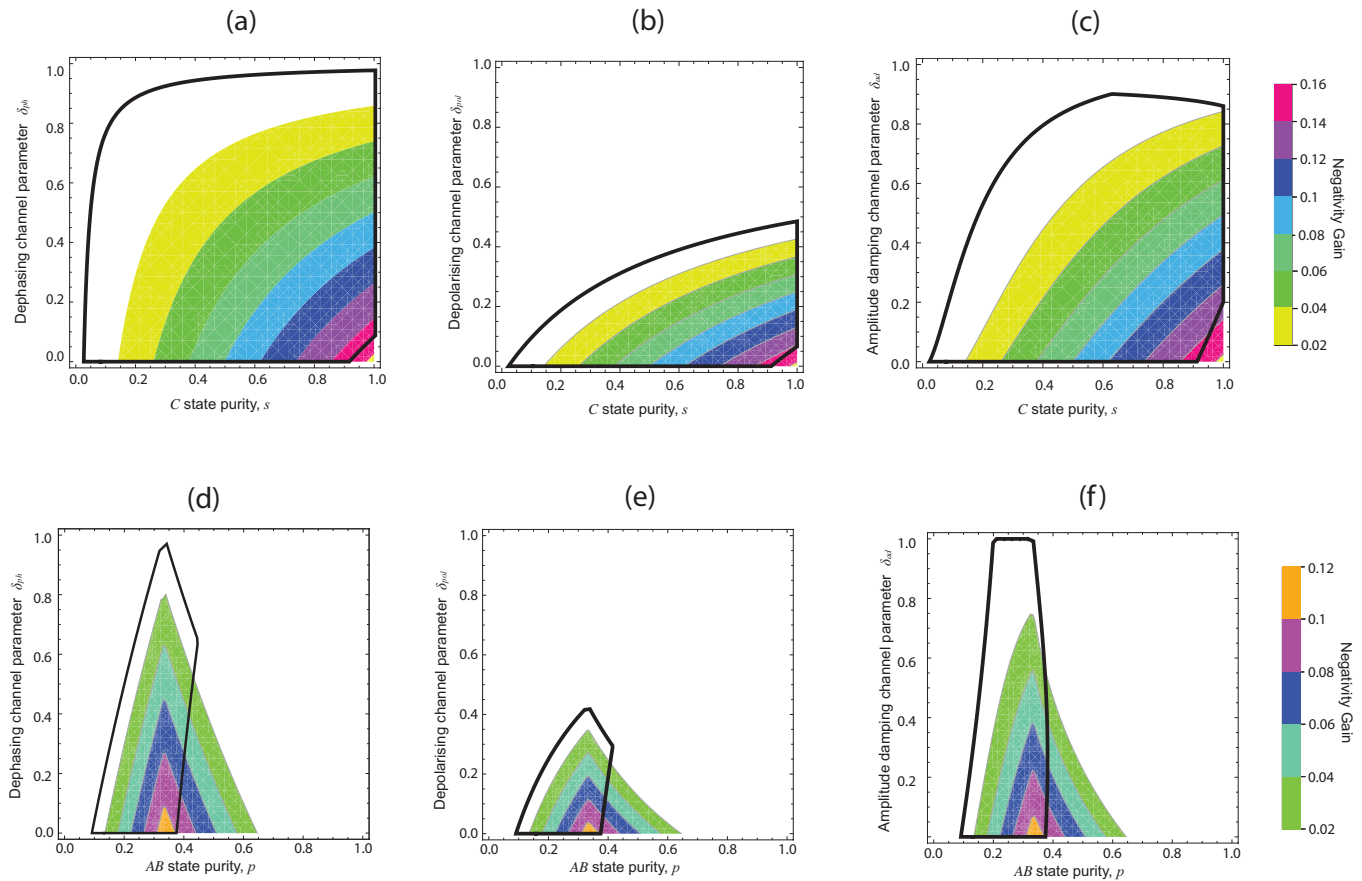}
\caption{(Color online) Entanglement gain in the indirect protocol of Fig.
\ref{p1}. Results for different noises are presented in columns.
The thick black lines include the regions in which the protocol is
excessive. The principal system begins in a Werner state, Eq.~\eqref{werner}.
The upper plots, (a) to (c), present the gain as a function of purity
of ancilla, $s$ in Eq.\eqref{ALPHAC}. For these plots entanglement
admixture in the Werner state is fixed to $p=0.34$, but essentially
the same qualitative results are obtained for other admixtures. The
optimal gain is reached for pure state of $C$ and lies outside the
excessive border for low noise parameters. \\
The lower plots, (d) to (f), present entanglement gain as a function
of purity of the principal system, $p$ in Eq.~\eqref{werner}. For
these plots $s=2/3$ but again results are similar for different values
of $s$. Notice that for all noises the optimal protocol is for $p=1/3$
(disentangled initial state) and it is always excessive (whenever
effective).}
\label{r1} 
\end{figure*}


\subsection{Indirect distribution via noisy channels}

\label{protocols}

In the numerical studies presented in this subsection we extend the examples
presented in Ref.~\cite{exp1}. Consider the scenario depicted in
Fig.~\ref{p1}. The particles $A$ and $B$ managed by Alice and Bob are already at their respective laboratories. 
We assume they are prepared in the Werner state 
\begin{equation}
\alpha_{AB}=p\ket{\phi_{+}}\bra{\phi_{+}}+\frac{(1-p)}{4}\openone_{4}\label{werner}
\end{equation}
where $\ket{\phi_{+}}=(\ket{00}+\ket{11})/\sqrt{2}$ is a Bell state ($\ket{0}$ and $\ket{1}$ are the eigenstates of
local $\sigma_{z}$ operator) and $\openone_{4}$ is the identity
matrix for a two-qubit system. For $p\le1/3$, the Werner state is separable. Particle $C$ starts in a state of
the form \cite{sep2}:
\begin{equation}
\alpha_{C}=\frac{1}{2}(\openone_{2}+s\sigma_{x}),\label{ALPHAC}
\end{equation}
where $s\in[0,1]$. We choose the interaction between $A$ and $C$ to be the controlled-phase (c-phase) gate, as in
Ref.~\cite{sep2}. Extensive numerical study suggests that this is
the optimal choice. The noisy mechanism is supposed to affect the channel that connects the remote laboratories, and thus acts on $C$ as it travels
from Alice to Bob. The c-phase interaction is assumed to take place at Alice's lab. We then calculate how much is the entanglement gain of such indirect protocol.

The results of our quantitative analysis are presented in Fig.~\ref{r1}. Panels (a)-(c) demonstrate
that the largest entanglement gain is achieved when particle $C$ is prepared in a pure state, independently of the type of noise in the
channel. Furthermore, this maximal gain is achieved via a non-excessive
protocol, again regardless of the applied noise. However, the parameter
regions where the protocol is non-excessive are very small. 
If the noise is not too weak, and if Alice is not able to put the ancilla $C$ in a very pure state, then there will be more entanglement gain than communicated. Notice that this is the case for almost all values of the noise parameters.
Complementary results are obtained if negativity gain is calculated
as a function of entanglement admixture in the Werner state, as shown in Fig.~\ref{r1} (d)-(f). This time the parameter range of excessive distribution is reduced, although it always contains the protocols that achieve maximum
gain. Note that, independent of the noisy channel and the strength
of noise, the largest negativity increment is obtained for initial
Werner states lying on the separability border, i.e. for $p=1/3$. To some
extent, this can be intuitively explained: the degree of entanglement in a
maximally entangled state cannot be improved and the increment of
entanglement of a ``deeply'' separable state has to first overcome
the distance to the separability border.

\subsection{Direct-then-indirect distribution via noisy channels}

\begin{figure}
\includegraphics[width=0.5\textwidth]{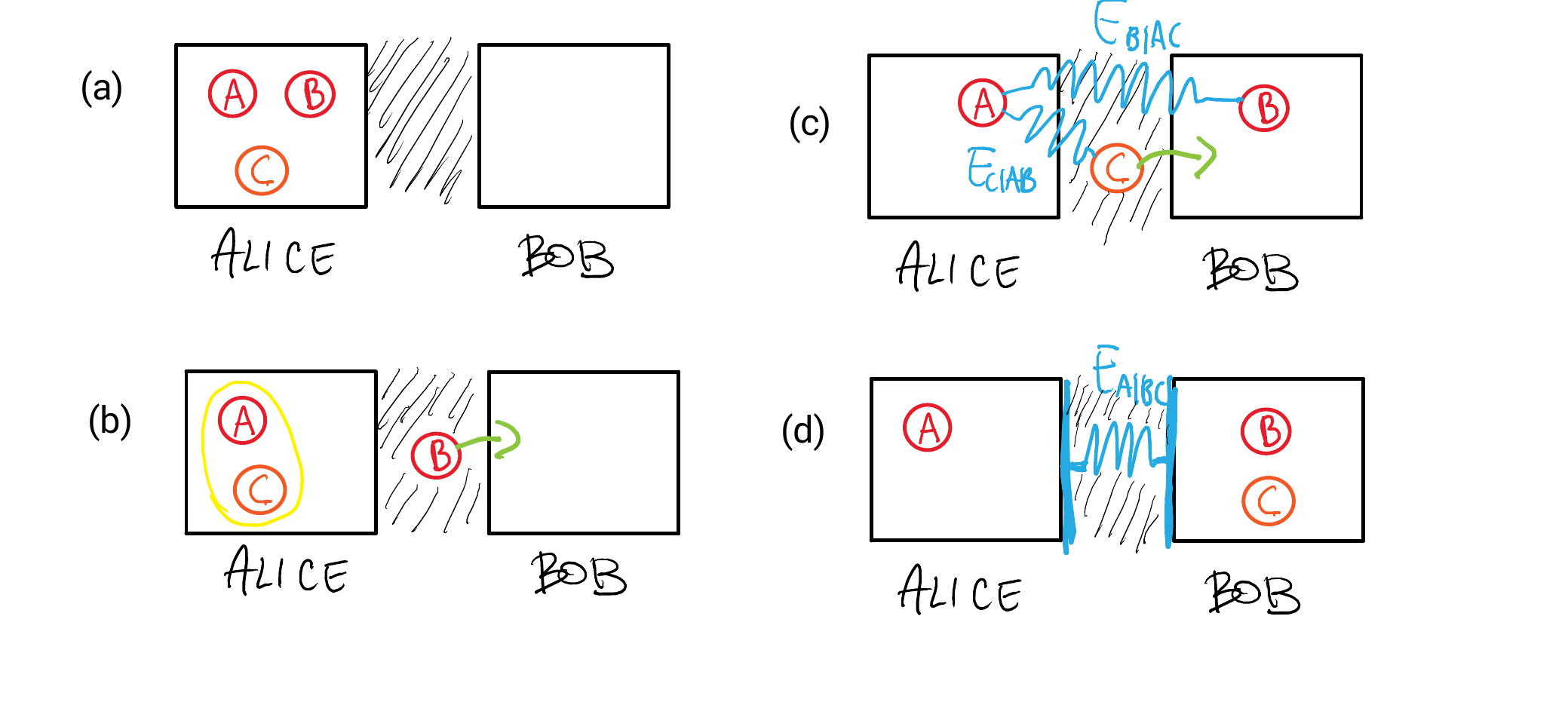}
\caption{(Color online) Direct-then-indirect distribution via noisy channel.
(a) Now $A$ and $B$ are initially correlated locally in Alice's
laboratory. (b) While $B$ travels through the channel (direct protocol),
$A$ and $C$ interact with each other. (c) Finally $C$ reaches Bob,
(d) changing the entanglement between the labs. }
\label{p2} 
\end{figure}

The protocol described in the previous subsection assumes that particles $A$ and $B$ are already far from each other, yet prepared in a (partially) entangled state. Clearly, this has a cost, as Alice and Bob should be able to generate
it using a noisy channel that cannot be bypassed in our scenario. In order to include the cost of preparing
such initial state, we now study the protocol illustrated in Fig.~\ref{p2}.
Alice, who initially holds both particles $A$ and $B$, prepares the Werner state in Eq.~(\ref{werner}). 
Next, particle $B$ is communicated via the noisy channel to Bob. In addition to such
direct protocol, Alice and Bob run the indirect protocol using
the initial state of the ancilla given in Eq.~(\ref{ALPHAC}) and the c-phase
gate, as illustrated before.

\begin{figure*}[t]
\includegraphics[width=1.1\textwidth]{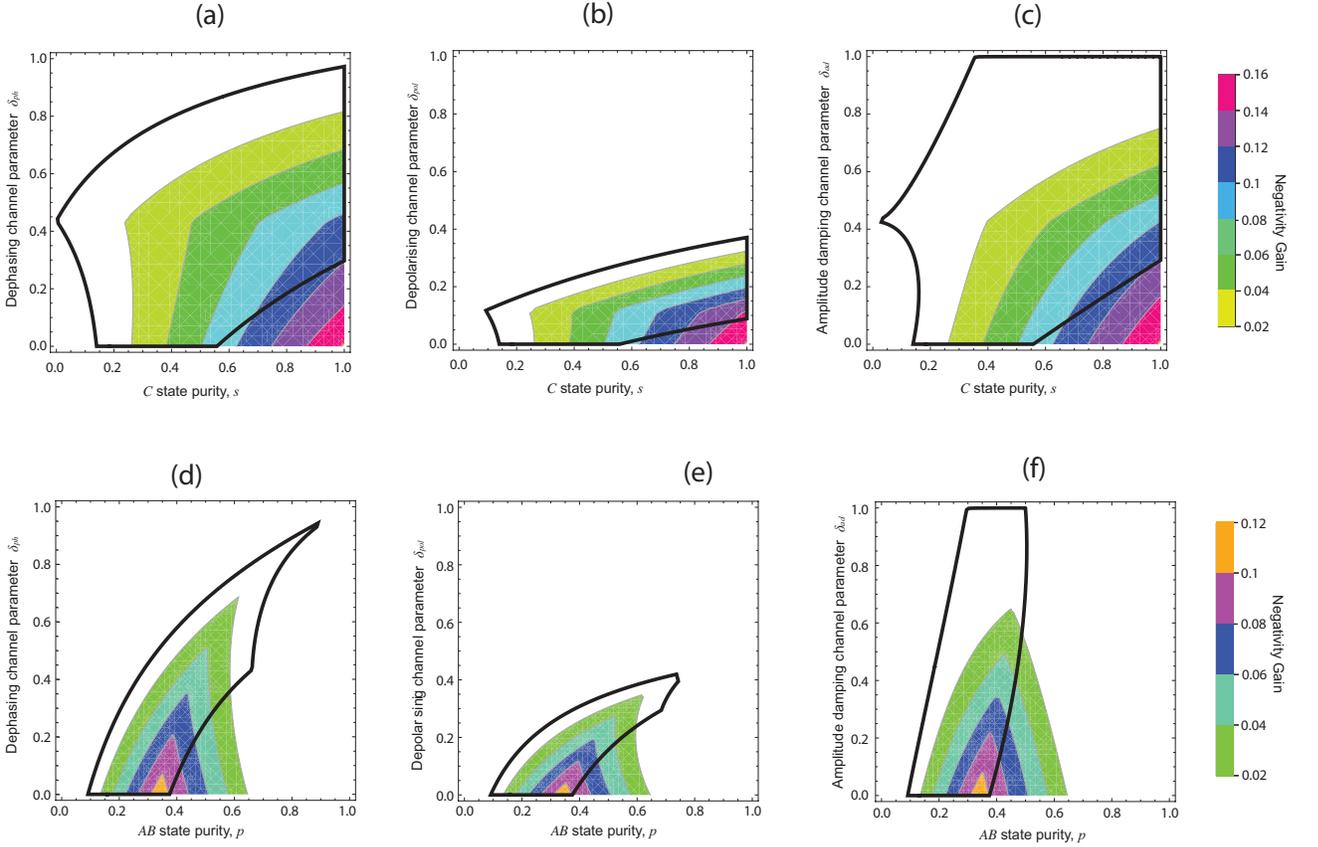}
\caption{(Color online) Entanglement gain achieved solely in the indirect protocol
of Fig. \ref{p2}, i.e., the final negativity minus the negativity
between the laboratories after the direct protocol. The region delimited by
the black line shows where the protocol is excessive. Results for
different noises are presented in columns. In the upper plots, (a)
to (c), the parameter of the initial Werner state is $p=0.34$ but
the same qualitative results are obtained for different values of $p$. Notice
that, similarly to Fig. \ref{r1}, the optimal gain is achieved with
$C$ in a pure state. In the lower plots, (d) to (f), the negativity
gain is presented as a function of the parameter of the Werner state
of $A$ and $B$, i.e., $p$ in Eq.~(\ref{werner}). Ancilla $C$
is initially in the state (\ref{ALPHAC}) with $s=2/3$. Again, qualitatively
the same results are seen for different $s$. In this case, the largest
gain is achieved by excessive protocols (whenever effective) but differently
from Fig. \ref{r1} the optimal value of $p$ now depends on the strength
of noise.}
\label{r2} 
\end{figure*}

In Fig.~\ref{r2} we present the negativity gain in the indirect protocol
alone, i.e. the entanglement achieved after the direct distribution step is
subtracted from the final entanglement between the laboratories. This
allows us to easily compare this situation with the entanglement gains presented
in Fig.~\ref{r1}. First we note that in both cases the range of
parameters giving excessive gain is very large. This can be quantified by the areas in the plots of Figs. \ref{r1} and \ref{r2} within the bold contour, and interpreted as high robustness of excessive protocols against noise. 

Other similarities to the protocols of the previous subsection relate to
the excessiveness of the optimal entanglement gain. As demonstrated
in Fig.~\ref{r2} (a)-(c) the highest gain is obtained starting with pure
ancilla states and the protocol is non-excessive, independently of
the type of noise. However, it should be stressed that extra communication
of particle $B$ over noisy channel leads to more mixed ancillary
states, while the indirect protocol is
still non-excessive. The second similarity is that the best protocols are always excessive [cf. Fig.~\ref{r2} (d)-(f)]. Summing up, whether the protocol is excessive or non-excessive depends on easily controllable parameters, such as the purity of the state of particle 
$C$ and parameter $p$ in the state of  the pair $AB$.

Not all the features of the present protocol are the same as in the one illustrated
in the previous subsection. The maximum gain of negativity now depends on both the 
purity of the Werner-state resource and the actual amount of noise in the channel.
For perfect channels, it is always best to begin with the Werner state
at the border of separability. For noisy channels, the entanglement admixture in the initial
 Werner state corresponding to the largest entanglement gain depends on the strength of the channel.
For stronger noise, i.e., larger $\delta_{\mathrm{ch}}$, one should
begin with larger entanglement in the Werner state as part of it would be
lost during the initial transmission of $B$. However, as intuition suggests,
such initial entanglement cannot be too large, as it is very difficult
to improve the degree of entanglement in states that are already highly entangled.

Finally, we note that starting with direct distribution of mixed Werner states, there is a range of noise parameters leading to no direct entanglement
gain although the channels are not entanglement breaking (see Sec.~\ref{SEC_EBC}).
This is clear for initially separable Werner states but also happens if they are weakly entangled.
If Alice can only prepare such Werner state, the indirect protocols
provide the only means of entanglement creation between the laboratories,
and for certain noise and state parameters, the excessive distribution
is the only viable option.

\begin{figure}
\includegraphics[width=0.5\textwidth]{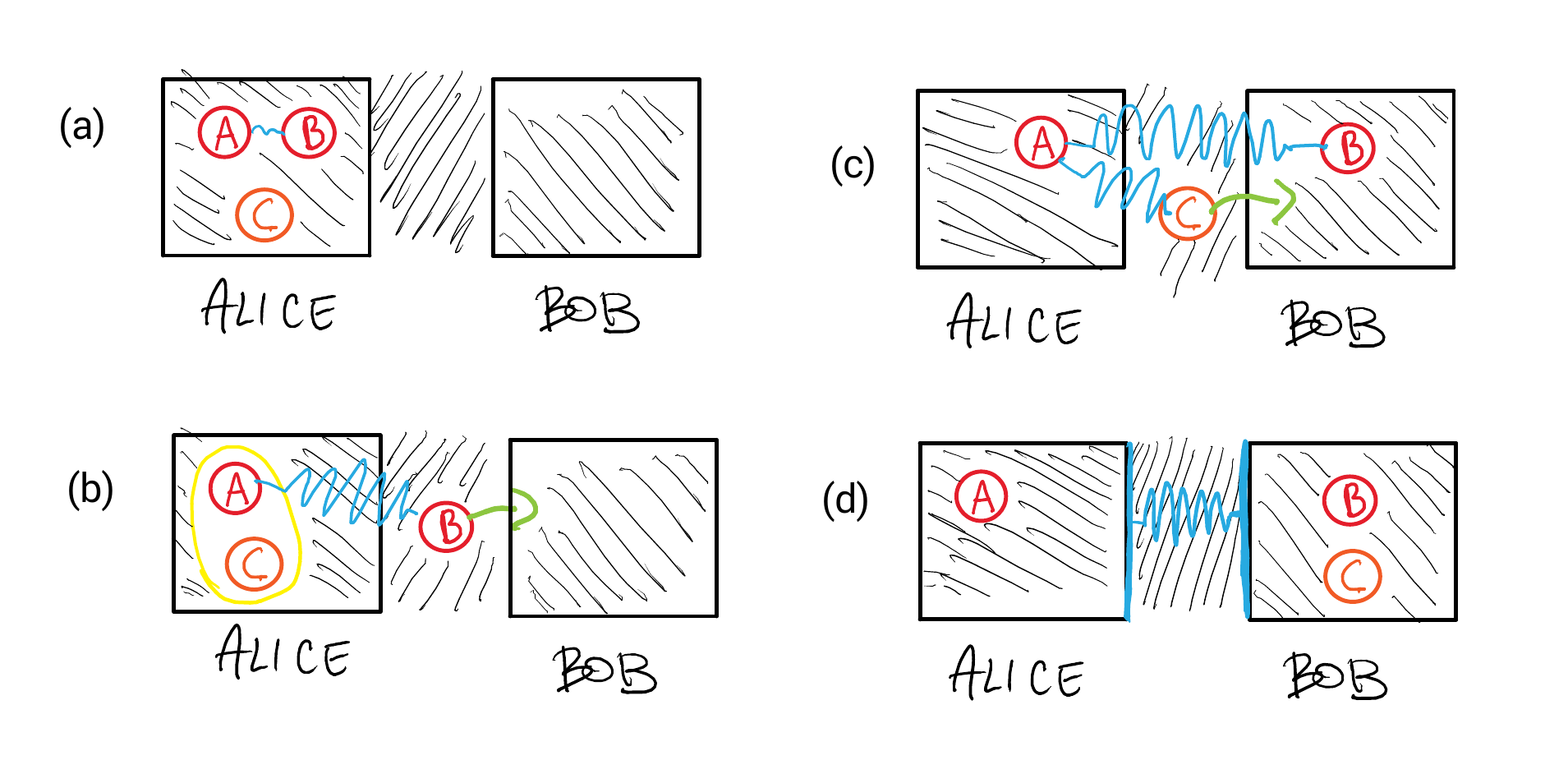}
\caption{(Color online) Entanglement distribution in the presence of local
noises as well as noisy channel. (a) All qubits are initially in Alice's
laboratory. (b) $B$ travels to reach Bob, while $A$ and $C$ stay
in the noisy laboratory. (c) Finally $C$ travels through the channel
while $A$ and $B$ are affected by their respective local noises.
(d) We compare the final entanglement between the laboratories with entanglement
between them after the direct protocol, i.e., panel (c).}
\label{p3} 
\end{figure}

\begin{figure*}
\includegraphics[width=1\textwidth]{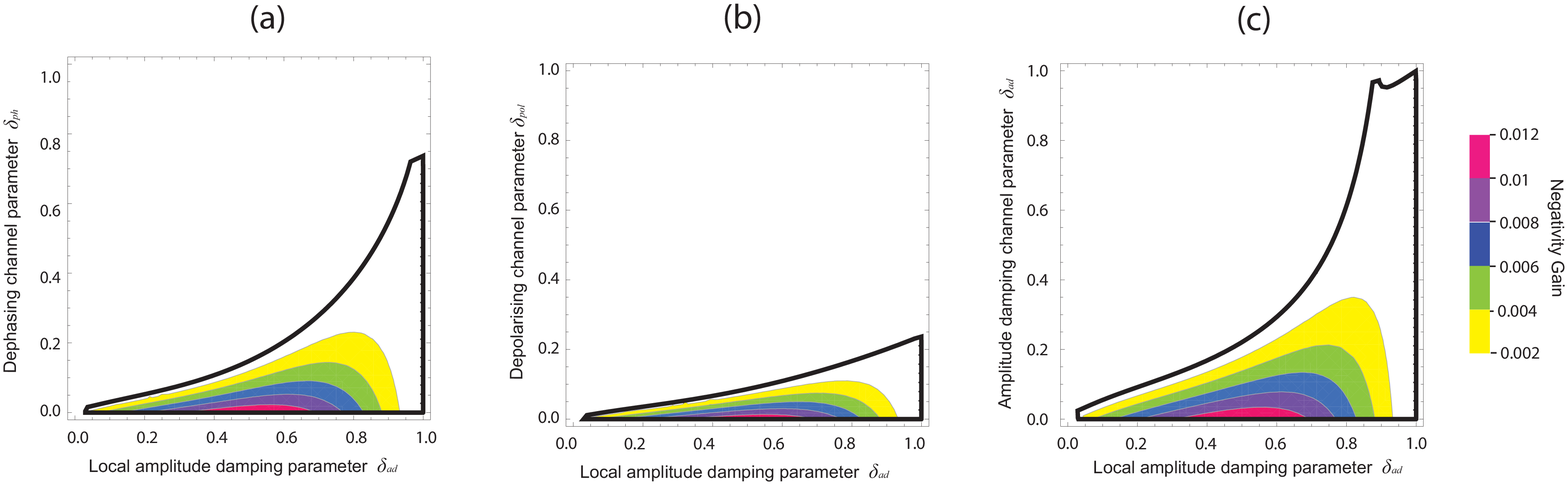}
\caption{(Color online) Entanglement gain achieved in the indirect protocol
described in Fig. \ref{p3}. The parameter for the initial Werner state is $p=0.34$ while $C$ is prepared in a pure state. The local noises of Alice and Bob are
chosen as amplitude damping noises and assumed here to have the same
strength, given in the horizontal axis. Along the vertical axes we
present the strength of different channel noises. Within the black
line the protocol is excessive, i.e., the communicated entanglement
is smaller than the gain. There are many pairs of noise parameters
that allow excessive protocols to take place. Notice that the points
of largest gain, for local amplitude damping parameter about $\delta_{\mathrm{ad}}\sim0.5$, are inside
the excessive region. The gain in negativity  is however much smaller than
in the absence of local noise, see Figs. \ref{r1} and \ref{r2}. }
\label{local} 
\end{figure*}

\subsection{Direct-then-indirect distribution with noisy laboratories and via
noisy channels}

Our last numerical example proceeds towards the consideration of more realistic experimental
situations. In addition to a noisy channel connecting them, we now also include local noise
affecting the laboratories of Alice and Bob, as Fig.~\ref{p3} illustrates.
The protocol goes as follows: Alice begins with three particles: $AB$
prepared in the Werner state of Eq.~(\ref{werner}), and $C$ prepared in state $(\ket{0}+\ket{1})/\sqrt{2}$. 
The Werner state could be regarded as an output of some non-perfect entangling procedure.
Next we assume that $A$ and $C$ interact via c-phase gate which
is instantaneous and ideal. Then particle $B$ travels to Bob via
the noisy channel whereas $A$ and $C$ are independently affected
by local noise. Finally, ancilla $C$ experiences the channel noise,
particle $A$ experienced the noise in Alice's laboratory and particle
$B$ experiences the noise in Bob's laboratory.

A representative case is when both local noises are amplitude damping,
e.g. thermal baths, while the noise in the channel can be different.
Fig. \ref{local} shows that in the present case the excessive protocols
are also very robust as characterised by the range of noise parameters
for which there is a gain in negativity. The protocols giving the
highest entanglement gain are always found to be excessive. Note also
that the gain is one order of magnitude smaller than the
cases without local noise. Altogether, although excessive protocols
seem unusual, they are actually quite relevant when entanglement is gained
via simple protocols operating in natural noisy conditions.

\subsection{No distribution via entanglement breaking channels}

\label{SEC_EBC}

The most intriguing feature of excessive protocols is the possibility
to vein entanglement even if no entanglement is communicated. The existence of 
such protocols is known since the work by Cubitt \emph{et al.}~\cite{cubitt}. A natural question arises in the context
of noisy channels: If in order to gain entanglement, no entanglement
has to cross the channel, could EB channels allow for excessive entanglement gain?

A positive answer to such question would provide a striking example
of usefulness of indirect and excessive protocols. Unfortunately, the
answer is negative, as can be seen from Eq.~(\ref{povm}). The same
effect as the action of EB channels could be obtained by the following local
operations and classical communication. Instead of sending $C$ via
the channel, Alice performs (in her laboratory) the POVM measurement
corresponding to the EB channel. She then sends the (classical) outcome to Bob's location, where an ancillary system $D$ is prepared in the corresponding states $\gamma_{n}$. Alice and Bob do not know what is the actual
outcome of the measurement. Therefore, they assign a density operator as in Eq.~(\ref{povm}) to particles $A$, $B$, and $D$.
As only classical information is transmitted, entanglement does not grow.

A more formal proof emphasising that entanglement gain is calculated
for different bipartitions is given by the following inequality
\begin{equation}
E_{A:BD}'-E_{AC:B}\le E_{AC:BD}'-E_{AC:BD}\le0.
\end{equation}
The inequality to the right means that local operations and classical
communication do not increase entanglement, where $E_{AC:BD}$ denotes
entanglement before the action of the channel and $E_{AC:BD}'$ is
after the channel. The left inequality is obtained from $E_{AC:B}=E_{AC:BD}$,
as initially Bob's ancilla is completely uncorrelated, and using $E_{A:BD}'\le E_{AC:BD}'$,
due to tracing out one subsystem.

The case of EB channels clearly illustrates that,
although communicated quantum discord is necessary for entanglement
gain~\cite{bounds1,bounds2}, it is not a tight bound to entanglement
gain in every distribution protocol~\cite{sep2}. The state resulting
from EB channels can possess some discord as measured on the ancilla, but this can be thought of as being locally created by a device in
Bob's laboratory, which in fact is fed with purely classical communication.

\section{Conclusions}

We have classified protocols for entanglement gain into direct
and indirect, depending on whether entanglement is generated via mutual
interaction or with the help of ancillas, and further into excessive
and non-excessive, depending on whether entanglement gain exceeds
the amount of communicated entanglement. Analytical examples were
provided illustrating the various protocols. Such analysis has been complemented by a 
numerical study showing the usefulness of excessive protocols in the presence of noise.

These results will be of use in quantum information science where
distributing entanglement is a prerequisite for many relevant tasks.
Achieving this in a cheap and reliable way is essential to the development
of future quantum technologies. They also reinforces the role of quantum
discord as a beneficial and practically relevant quantity~\cite{dakic,PhysRevLett.112.140507, gu2, gu3, bobby, datta, lanyon, girolami, streltsov2011linking, streltsov2013quantum, PhysRevA.83.032323, horodecki2012quantum, PhysRevLett.100.090502, PhysRevLett.106.220403,pirandola2014quantum}: as discord appears to bound the entanglement gain and to allow for excessive entanglement distribution, it emerges as a key player of fundamental relevance in the 
context of quantum communication and networking.

\acknowledgments
MZ is grateful to the Centre for Theoretical Atomic, Molecular, and Optical Physics, Queen's University Belfast, 
for hospitality at various stages in the development of this project. This work is supported by the National Research Foundation, Ministry
of Education of Singapore Grant No. RG98/13, and start-up grant of
the Nanyang Technological University.
SB thanks CQT, NUS for supporting a visit when part of the work reported was done. 
SB is supported in part by DST-SERB project SR/S2/LOP-18/2012.
AB is supported by DST-SERB project SR/S2/LOP-18/2012. 
SH is supported by a fellowship from CSIR, Govt. of India. 
PD is supported by an Inspire Fellowship from DST, Govt. of India.
TK wishes to acknowledge the funding support for this project from Nanyang Technological University under the Undergraduate Research Experience on CAmpus (URECA) programme. MP acknowledges the John Templeton Foundation (grant ID 43467), and the UK EPSRC (EP/M003019/1) for financial support. 
\bibliographystyle{apsrev4-1}
\bibliography{indirect.bib}

\end{document}